\documentclass[a4paper,envcountsame,envcountsect,runningheads]{llncs}

\pdfoutput=1

\usepackage{microtype}
\usepackage{amsmath, amssymb}
\usepackage[mathscr]{eucal}
\usepackage{extarrows}
\usepackage{ifthen}
\usepackage{enumitem}
\usepackage{tabularx}
\usepackage[svgnames, hyperref]{xcolor}

\usepackage[disable,textsize=small]{todonotes}

\newcommand{\gsnote}[1]{\todo[backgroundcolor=green!20]{[GS] #1}}

\usepackage{hyperref}
\hypersetup{
    colorlinks=true,    %
    linkcolor=blue,     %
    citecolor=blue,     %
    filecolor=magenta,  %
    urlcolor=blue,      %
}

\usepackage[capitalise,noabbrev,nameinlink]{cleveref}

\usepackage{tikz}
\usetikzlibrary{arrows, calc, shapes, snakes, automata, fit, positioning, intersections}
\tikzset{>=stealth'} %

\tikzstyle{every picture} = [style=semithick]
\tikzstyle{every node}    = [font=\small]
\tikzstyle{every state}   = [thick, minimum size=1mm, inner sep=2pt]

\usepackage[inference]{semantic}

\spnewtheorem{fact}[theorem]{Fact}{\itshape}{\rmfamily}
\crefname{fact}{Fact}{Facts}
\spnewtheorem*{assumption}{Assumption}{\itshape}{\rmfamily}
\crefname{assumption}{Assumption}{Assumptions}

\spnewtheorem{observation}[theorem]{Observation}{\itshape}{\rmfamily}
\crefname{observation}{Observation}{Observation}

\usepackage{thmtools,thm-restate}

\usepackage{algpseudocode}

\usepackage{macros}

\newcommand{\parname}[1]{\paragraph{#1.}}

\title{%
  On Boundedness Problems for Pushdown Vector Addition Systems%
  \thanks{%
    This work was partially supported by
    ANR project \textsc{ReacHard} (ANR-11-BS02-001).
  }
}

\author{%
    J\'{e}r\^{o}me Leroux\inst{1}
  \and
  Gr\'{e}goire Sutre\inst{1}
  \and
  Patrick Totzke\inst{2}
}

\institute{%
  Univ. Bordeaux \& CNRS, LaBRI, UMR 5800, Talence, France
  \and
  Department of Computer Science, University of Warwick, UK
}

\begin{document}

\maketitle

\begin{abstract}
  We study pushdown vector addition systems,
which are synchronized products of pushdown automata with vector
addition systems.
The question of the boundedness of the reachability set for this
model
can be refined into two decision problems that ask
if infinitely many counter values or stack configurations
are reachable, respectively.
Counter boundedness seems to be the more intricate problem.
We show decidability in exponential time
for one-dimensional systems.
The proof is via a small witness property derived
from an analysis of derivation trees of
grammar-controlled vector addition systems.

\end{abstract}

\section{Introduction}

\label{sec:pdvass}

Pushdown vector addition systems are
finite automata that can independently manipulate a pushdown stack and
several counters. They are defined as
synchronized products of vector addition systems with pushdown automata.
Vector addition systems, shortly \emph{VAS}, are a classical model for
concurrent systems and are  %
computationally equivalent to
Petri nets.
Formally, %
a $k$-dimensional \emph{vector addition system} is a finite set
$\vec{A}\subseteq \setZ^k$ of vectors called \emph{actions}.
Each action $a\in\vec{A}$ induces
a binary relation
$\vstep{\vec{a}}$ over $\N^k$, %
defined by
$\vec{c} \vstep{\vec{a}} \vec{d}$ if $\vec{d} = \vec{c} + \vec{a}$.

A $k$-dimensional \emph{pushdown vector addition system}, shortly \emph{PVAS},
is a tuple $(Q, \Gamma, q_\init, \vec{c}_\init, w_\init, \Delta)$ where
$Q$ is a finite set of \emph{states},
$\Gamma$ is a finite \emph{stack alphabet},
$q_\init \in Q$ is an \emph{initial state},
$\vec{c}_\init \in \setN^k$ is an \emph{initial assignment of the counters},
$w_\init \in \Gamma^*$ is an \emph{initial stack content}, and
$\Delta\subseteq Q\x\setZ^k\times \Op(\Gamma)\x Q$ is a finite set of
\emph{transitions} where $\Op(\Gamma)\eqdef\{\push(\gamma),\pop(\gamma),\nop \mid
\gamma\in\Gamma\}$ is the set of stack operations.
The \emph{size} of VAS, PVAS (and GVAS introduced later) are defined as expected with numbers encoded in binary.

\begin{figure}[h]
\begin{minipage}[c]{.5\textwidth}
  \centering
  \begin{algorithmic}[1]
    \State $x \gets n$
    \Procedure{DoubleX}{}
    \If {$(\star~\wedge~x > 0)$}
        \State $x\gets (x -1)$
        \State \Call{DoubleX}{}
    \EndIf
    \State $x \gets (x +2)$
    \EndProcedure{}
    \end{algorithmic}
  \end{minipage}
  \begin{minipage}[c]{.5\textwidth}
  \centering
  \begin{tikzpicture}[node distance=1.2cm]
      \node[state,initial] (2) {$2$};
        \node[state] (3) [below of=2] {$3$};

        \node[state] (5) [below of=3] {$5$};
        \node[state] (6) [right of=2,xshift=0.5cm] {$6$};
        \node[state] (7) [below of=6] {$7$};
        \node[state] (8) [below of=7] {$8$};

        \draw [->] (2) to[] node {} (3);

        \draw [->] (3) to[] node[right] {$-1$} (5);
        \draw [->] (6) to[] node {} (7);

        \draw [->] (3) to[] node {} (7);
        \draw [->] (5) to[bend left=60] node[left] {$\push(A)$} (2);
        \draw [->] (7) to[] node[left]{$+2$} (8);
        \draw [->] (8) to[bend right=60] node[right] {$\pop(A)$} (6);
    \end{tikzpicture}
  \end{minipage}
  \caption{A PVAS modeling a recursive program.\label{fig:program}}
\end{figure}
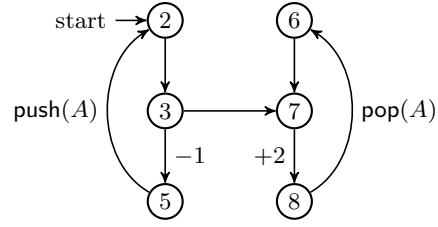

\begin{example}\label{ex:program}
  Consider the program on the left of \cref{fig:program},
  that doubles the value of the global variable $x$.
  The $\star$ expression non-deterministically evaluates to a Boolean,
  as it is often the case in abstraction of programs~\cite{DBLP:conf/pldi/BallMMR01}.
  On the right is a $1$-dimensional PVAS that
  models this procedure: states correspond to lines in the program
  code, operations on the variable $x$ are directly applied, and the call stack
  is reflected on the pushdown stack.
  \qed
\end{example}

The semantics of PVAS is defined as follows.
A \emph{configuration} is a triple $(q,\vec{c},w)\in Q\x\N^k\x\Gamma^*$
consisting of a state, a vector of natural numbers, and a stack content.
The binary \emph{step} relation $\rightarrow$ over configurations is defined by
$(p,\vec{c},u)\rightarrow(q,\vec{d},v)$ if there is a transition
$(p,\op,\vec{a},q)\in\Delta$ such that
$\vec{c}\vstep{\vec{a}}\vec{d}$ and one of the following conditions holds:
either
$\op=\push(\gamma)$ and
$v=u\gamma$,
or
$\op=\pop(\gamma)$ and $u=v\gamma$,
or
$\op=\nop$ and $u=v$.
The reflexive and transitive closure of $\rightarrow$ is denoted by $\xrightarrow{*}$.

\smallskip

The \emph{reachability set} of a PVAS is the set of configurations
$(q,\vec{c},w)$ such that
$(q_\init,\vec{c}_\init,w_\init)\xrightarrow{*}(q,\vec{c},w)$.
The reachability problem asks if a given
configuration $(q,\vec{c},w)$ is in the reachability set
of a given PVAS.
The decidability of
this problem is open. Notice that for vector addition systems, even though the reachability
problem is decidable~\cite{May1981,Kos1982}, no primitive upper bound of complexity is known
(see~\cite{LerouxSchmitz:2015:LICS} for a first upper bound).
However, a variant called the
coverability problem is known to be \EXPSPACE-complete~\cite{Rac1978,Lip1976}. Adapted to
PVAS, the coverability problem takes as input a PVAS
and a state $q\in Q$ and asks if there exists a reachable
configuration of the form $(q,\vec{c},w)$ for some $\vec{c}$ and $w$.
The decidability of the coverability problem for PVAS is also open.
In fact,
coverability and reachability are inter-reducible (in logspace) for
this class~\cite{Laz2013,LerouxSutreTotzke:2015:ICALP}.
In dimension one, we recently proved that coverability is
decidable~\cite{LerouxSutreTotzke:2015:ICALP}.

\smallskip

Both coverability and reachability are clearly decidable for
PVAS with finite reachability sets. These PVAS are
said to be \emph{bounded}.
In~\cite{LPS2014}, this class is
proved to be recursive, i.e. the boundedness problem for PVAS
is decidable. The complexity of this problem is known to be
{\TOWER}-hard~\cite{Laz2013}. The decidability is obtained by observing that if the reachability
set of a PVAS is finite, its cardinality is at most
hyper-Ackermannian in the size of the PVAS. Even though
this bound is tight~\cite{LPS2014}, the exact
complexity of the boundedness problem is still open. Indeed, it is
possible that there exist small certificates that witness infinite reachability
sets. For instance, in the VAS case,
the reachability set can be finite and Ackermannian.
But when it is infinite, there
exist small witnesses of this fact~\cite{Rac1978}.
This yields an optimal~\cite{Lip1976} exponential-space
algorithm for the VAS boundedness problem. Extending this
technique to PVAS is a challenging problem.

\smallskip
The boundedness problem for PVAS can be refined in two
different ways. In fact, the infiniteness of the reachability set may come from
the stack or the counters.
We say that a PVAS is \emph{counter-bounded} if
the set of vectors $\vec{c}\in\setN^k$ such that $(q,\vec{c},w)$ is
reachable for some $q$ and $w$, is finite. Symmetrically, a PVAS is called
\emph{stack-bounded} if the set of words $w\in \Gamma^*$ such that
$(q,\vec{c},w)$ is reachable for some $q$ and $\vec{c}$, is finite.
The
following lemma shows that the two associated decision problems are
at least as hard as the boundedness problem.

\begin{restatable}{lemma}{LemBoundednessReduction}
    \label{lem:BoundednessReduction}
  The boundedness problem is reducible in logarithmic space to the
  counter-boundedness problem and to the stack-boundedness problem (the
  dimension $k$ is unchanged by the reduction).
\end{restatable}

The stack-boundedness problem can be solved by adapting the algorithm
introduced in~\cite{LPS2014}
for the PVAS boundedness problem.
Informally, this algorithm explores
the reachability tree and stops as
soon as it detects a cycle of transitions whose iteration
produces infinitely many reachable configurations.
If this cycle increases the stack, we can immediately conclude
stack-unboundedness.
Otherwise, at least one counter can be increased to
an arbitrary large number. By replacing the value of this counter by
$\omega$ and then resuming the computation of the tree
from the new (extended) configuration, we obtain a Karp\&Miller-like
algorithm~\cite{DBLP:journals/jcss/KarpM69} deciding
the stack-boundedness problem. We deduce the following
result.
\begin{lemma}
  The stack-boundedness problem for PVAS is decidable.
\end{lemma}

Concerning the counter-boundedness problem, adapting the algorithm
introduced in~\cite{LPS2014} in a similar way
seems to be more involved. Indeed, if we detect a cycle that
only increases the stack, we can iterate it and represent its effect
with a regular language. However, we do not know how to effectively truncate the
resulting tree to obtain an algorithm deciding the counter-boundedness problem.

\parname{Contributions}
In this
paper we solve the counter-boundedness problem for the special case of
dimension one. %
We show that in a grammar setting,
PVAS counter-boundedness corresponds
to the boundedness problem for 
prefix-closed, grammar-controlled vector addition systems.
We show that in dimension one, this problem is decidable
in exponential time.
Our proof is based on
the existence of small witnesses exhibiting the unboundedness
property. This complexity
result improves the best known upper bound for the classical
boundedness problem for PVAS in dimension one. In fact, as
shown by the following \cref{ex:ack}, the reachability set of a bounded
$1$-dimensional PVAS can be Ackermannian large.
In particular,
the worst-case running time of the algorithm introduced in~\cite{LPS2014}
for solving the boundedness problem is at least Ackermannian even in
dimension one.

\begin{figure}[h]
  \centering
  \begin{tikzpicture}[text centered,->,bend angle=23]
      \node[state] (init) {$\perp$};
      \node[state] (0) [above of=init, node distance=1.75cm] {$0$};
      \node[state] (m) [right of=init, node distance=3.5cm]  {$m$};
      \node[state] (1) [left  of=init, node distance=3.5cm]  {$1$};

      \draw (init) to[bend left] node[left, pos=0.7] {$\pop(\gamma_0)$} (0);
      \draw (0)    to[bend left] node[right,pos=0.3] {$1$} (init);

      \draw (init) to[bend right] node[above, yshift=-0.5mm] {$\pop(\gamma_1)$} (1);
      \draw (1)    to[bend right] node[below] {$1$} node[above] {$\push(\gamma_{0})$} (init);

      \draw (init) to[bend left] node[above, yshift=-0.5mm] {$\pop(\gamma_m)$} (m);
      \draw (m)    to[bend left] node[below] {$1$} node[above, yshift=0.25mm] {$\push(\gamma_{m-1})$} (init);

      \draw[-, loosely dotted, bend angle=40, bend right] (1) to (m);

      \draw (1) to[loop left,looseness=10,out=150,in=-150] node[left]
      {$\begin{array}{@{}c@{}}\push(\gamma_{0})\\-1\end{array}$} (1);

      \draw (m) to[loop right,looseness=10,out=30,in=-30] node[right]
      {$\begin{array}{@{}c@{}}\push(\gamma_{m-1})\\-1\end{array}$} (m);
  \end{tikzpicture}
  \caption{One-dimensional PVAS that weakly compute Ackermann functions.\label{fig:ack}}
\end{figure}
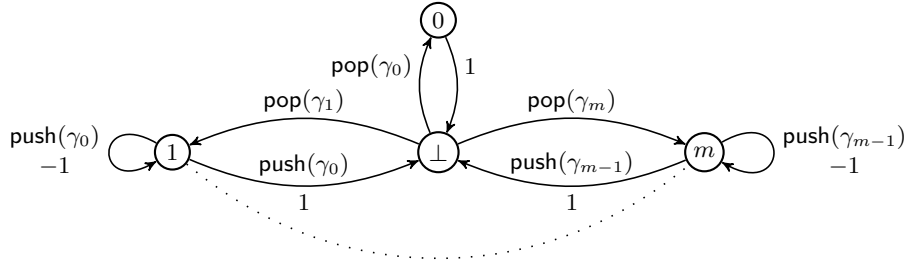

\begin{example}\label{ex:ack}

  The Ackermann functions $A_m:\setN\to\setN$, for $m\in\setN$, are defined by
  induction for every $n\in\setN$ by:
  $$A_m(n)\ \eqdef \ 
  \begin{cases}
    n+1 & \text{ if }m=0\\
    A_{m-1}^{n+1}(1) & \text{ if }m>0\\
  \end{cases}
  $$
  These functions are \emph{weakly computable} by the (family of) PVAS depicted in \cref{fig:ack},
  in the sense that:
  \begin{equation}
    \label{eq:ack-pvas}
    A_m(n) \ = \ 
    \max\{c \mid (\perp,n,\gamma_m)\xrightarrow{*}(\perp,c,\varepsilon)\}
  \end{equation}
  for every $m, n\in\setN$.
  Indeed, an immediate induction on $k\in\{0,\ldots,m\}$ shows that
  $(\perp,c,\gamma_k)\xrightarrow{*}(\perp,A_k(c),\varepsilon)$
  for every $c\in\setN$.
  For the converse inequality,
  let us introduce, for each configuration $(\perp,c,w)$,
  the number $\theta(c,w)$ defined by
  $$
  \theta(c, \gamma_{i_1}\cdots\gamma_{i_k})
  \ \eqdef\ 
  A_{i_1}\circ\cdots\circ A_{i_k}(c)
  $$
  An immediate induction on the number of times a run come back to the
  state $\perp$ shows that $(\perp,c,w)\xrightarrow{*}(\perp,c',w')$
  implies $\theta(c,w) \geq \theta(c',w')$.
  Since $\theta(c,\varepsilon) = c$,
  we derive that $A_m(n) \geq c$ for every $c$ such that
  $(\perp,n,\gamma_m)\xrightarrow{*}(\perp,c,\varepsilon)$.
  This concludes the proof of \cref{eq:ack-pvas}.

  \smallskip

  Notice that the reachability set of this PVAS is finite
  for any initial configuration.
  Indeed,
  $(\perp,c,w)\xrightarrow{*}(\perp,c',w')$
  implies $\theta(c,w) \geq \theta(c',w') \geq c' + \len{w'}$.
  Therefore, there are only finitely many reachable configurations in state $\perp$.
  It follows that the same property holds for the other states.
  \qed
\end{example}

\parname{Outline}
We recall some necessary notations about context-free
grammars and parse trees in the next section.
In \Cref{sec:model}, we present
the model of grammar-controlled vector addition
systems (GVAS) as previously introduced in~\cite{LerouxSutreTotzke:2015:ICALP},
and reduce the counter boundedness problem for PVAS
to the boundedness problem for the subclass of \emph{prefix-closed} GVAS.
We show in \cref{sec:certificates} that unbounded systems
exhibit certificates of a certain form.
\Cref{sec:struct} proves a technical lemma used later on
and finally,
in \cref{sec:small_certs},
we bound the size of minimal certificates %
and derive the claimed exponential-time upper bound.

\section{Preliminaries}
\label{sec:preliminaries}
We let $\setZbar \eqdef \setZ \cup \{-\infty, +\infty\}$ denote the
extended integers, and we use the standard extensions
of $+$ and $\leq$ to $\setZbar$.
Recall that $(\setZbar, \leq)$ is a complete lattice.

\parname{Words}
Let $A^*$ be the set of all finite \emph{words} over the alphabet $A$.
The \emph{empty word} is denoted by $\eps$.
We write $\len{w}$ for the \emph{length} of a word $w$ in $A^*$
and $w^k \eqdef w w \cdots w$ for its $k$-fold concatenation.
The \emph{prefix} partial order $\prefix$ over words
is defined by $u \prefix v$ if $v=uw$ for some word $w$.
We write $u\pprefix v$ if $u$ is a proper prefix of $v$.
A \emph{language} is a subset $L\subseteq A^*$.
A language $L$ is said to be \emph{prefix-closed} if
$u \prefix v$ and $v \in L$ implies $u \in L$.
\parname{Trees}
A \emph{tree} $T$ is a finite, non-empty, prefix-closed subset of $\setN^*$ satisfying the
property that if $tj$ is in $T$ then $ti$ in $T$ for all $i<j$.
Elements of $T$ are called \emph{nodes}.
Its \emph{root} is the empty word $\eps$.
An \emph{ancestor} of a node $t$ is a prefix $s \prefix t$.
A \emph{child} of a node $t$ in $T$ is a node $tj$ in $T$ with $j$ in
$\setN$.
A node is called a \emph{leaf} if it has no child (i.e., $t0 \not\in T$),
and is said to be \emph{internal} otherwise.
The \emph{size} of a tree $T$ is its cardinal $\card{T}$,
its \emph{height} is the maximal length $|t|$ of its nodes $t\in T$.
We let $\lex$ denote the lexicographic order on words in
$\setN^*$.

\parname{Context-free Grammars}
A \emph{context-free grammar} is a quadruple $G=(V,A,R,S)$,
where $V$ and $A$ are disjoint finite sets of
\emph{nonterminal} and \emph{terminal} symbols,
$S \in V$ is a \emph{start symbol},
and $R\subseteq V\times (V\cup A)^*$ is a finite set of
\emph{production rules}.
We write
\begin{equation*}
  X \pstep \alpha_1 \mid \alpha_2 \mid \ldots \mid \alpha_k
\end{equation*}
to denote that $(X,\alpha_1),\ldots,(X,\alpha_k) \in R$.
For all words
$w,w' \in (V\cup A)^*$,
the grammar admits a \emph{derivation step}
$w\gstep{}w'$
if there exist two words $u,v$ in $(V\cup A)^*$ and a
production rule $(X,\alpha)$ in $R$
such that $w= uXv$ and $w'=u\alpha v$.
Let $\gstep{*}$ denote the reflexive and transitive closure of
$\gstep{}$.
The \emph{language} of a word $w$ in $(V\cup A)^*$
is the set $\lang[^G]{w} \eqdef \{z\in A^*\mid w\gstep{*}z\}$.
The \emph{language} of $G$ is defined as $\lang[^G]{S}$, and it is
denoted by $\lang[^G]{}$.
A nonterminal $X\in V$ is called \emph{productive} if $\lang[^G]{X}\neq\emptyset$.
A context-free grammar $G=(V,A,R,S)$ is in \emph{Chomsky normal form}%
\footnote{To simplify the presentation, we consider a weaker normal
  form than the classical one, as we allow to reuse the start symbol.} 
if, for every production rule $(X, \alpha)$ in $R$,
either $(X, \alpha) = (S, \varepsilon)$ or $\alpha \in V^2 \cup A$.

\parname{Parse Trees}
A \emph{parse tree} for a context-free grammar $G=(V,A,R,S)$
is a tree $T$ equipped with a labeling function
$\lsymoperator:T\to (V \cup A \cup \{\varepsilon\})$
such that the root is labeled by $\lsym{\varepsilon} = S$
and $R$ contains the production rule
$\lsym{t} \pstep \lsym{t0} \cdots \lsym{tk}$ for every internal node $t$
with children $t0, \ldots, tk$.
In addition,
each leaf $t\not=\varepsilon$ with $\lsym{t} = \varepsilon$ is the only child of its parent.
Notice that $\lsym{t} \in V$ for every internal node $t$.
A parse tree is called \emph{complete} when
$\lsym{t} \in (A \cup \{\varepsilon\})$ for every leaf $t$.
The \emph{yield} of a parse tree $(T, \lsymoperator)$ is the word
$\lsym{t_1} \cdots \lsym{t_\ell}$ where
$t_1, \ldots, t_\ell$ are the leaves of $T$ in lexicographic order
(informally, from left to right).
Observe that for every word $w$ in $(V\cup A)^*$,
it holds that $S \gstep{*} w$ if, and only if,
$w$ is the yield of some parse tree.

\section{Grammar-Controlled Vector Addition Systems}
\label{sec:model}

In this section
we recall the notion
of GVAS
from
\cite{LerouxSutreTotzke:2015:ICALP}
and show that the
boundedness
problem for the subclass of \emph{prefix-closed} GVAS
is inter-reducible
to the counter-boundedness problem for
pushdown vector addition systems.

\begin{definition}[GVAS]
  A $k$-dimensional \emph{grammar-controlled vector addition system}
  (shortly, \emph{GVAS})
  is a tuple
  $G=(V,\vec{A},R,S,\vec{c}_\init)$ where
  $(V,\vec{A},R,S)$ is a context-free grammar,
  $\vec{A} \subseteq \Z^k$ is a VAS,
  and $\vec{c}_\init\in\setN^k$ is an initial vector.
\end{definition}

The semantics of GVAS is given by extending the relations
$\vstep{\vec{a}}$ of ordinary VAS %
to words
over $V\cup \vec{A}$ as follows.
Define $\vstep{\eps}$ to be the identity on $\N^k$ and let
$\vstep{z \vec{a}} {\eqdef} \vstep{\vec{a}} \circ \vstep{z}$
for $z \in \vec{A}^*$ and $\vec{a} \in \vec{A}$.
Finally,
let $\vstep{w} {\eqdef} \,\bigcup_{z \in \lang[^G]{w}} \vstep{z}$
for $w\in (V\cup \vec{A})^*$.
For a word $z=\vec{a}_1\vec{a}_2\cdots\vec{a}_n\in\vec{A}^*$ over the terminals,
we shortly write
$\sum z$ for the sum
$\sum_{i=1}^n a_i$.
Observe that
$\vec{c}\vstep{z}\vec{d}$ implies $\vec{d}-\vec{c}=\sum z$.

\medskip
Ultimately, we are interested in the relation $\vstep{S}$,
that describes the reachability relation via sequences of actions
in $\lang[^G]{S}$, i.e., those that are derivable from the starting symbol $S$
in the underlying grammar.
A vector $\vec{d} \in \setN^k$ is called \emph{reachable} from
a vector $\vec{c} \in \setN^k$ if $\vec{c} \vstep{S} \vec{d}$.
The
\emph{reachability set} of a GVAS is the set of vectors reachable from
$\vec{c}_\init$.

\medskip
A GVAS is said to be \emph{bounded} if its reachability set is
finite. The associated boundedness problem for GVAS is
challenging since the coverability
problem for PVAS,
whose decidability is still open,
is logspace reducible to it. However,
the various boundedness properties that we investigate on PVAS (see \cref{sec:pdvass})
consider all reachable configurations, without any acceptance condition.
So they intrinsically correspond to context-free languages
that are prefix-closed.
It is therefore natural to consider the same restriction for GVAS.
Formally,
we call a GVAS $G=(V,\vec{A},R,S,\vec{c}_\init)$ \emph{prefix-closed} when
the language $\lang[^G]{S}$ is prefix-closed.
Concerning the counter-boundedness problem for PVAS, the following
lemma shows that it is sufficient to consider the special case of prefix-closed GVAS.

\begin{restatable}{lemma}{LemPVAStoGVAS}
    \label{lem:PVAS_to_GVAS}
  The counter-boundedness problem for PVAS is logspace inter-reducible with
  the prefix-closed GVAS boundedness problem
  (the dimension $k$ is unchanged by both reductions).
\end{restatable}

In this paper,
we focus on the counter-boundedness problem for PVAS of dimension one.
We show that this problem is decidable in exponential time.
The proof is by reduction, using \cref{lem:PVAS_to_GVAS},
to the boundedness problem for prefix-closed $1$-dimensional GVAS.
Our main technical contribution is the following result.

\begin{theorem}
    \label{thm:main}
    The prefix-closed $1$-dimensional GVAS boundedness problem is decidable in
    exponential time.
\end{theorem}

For the remainder of the paper,
we restrict our attention to the dimension one,
and shortly write GVAS instead of $1$-dimensional GVAS.

\begin{example}
    \label{ex:gvas}
  Consider again the Ackermann functions $A_m$ introduced in \cref{ex:ack}.
  These can be expressed by the GVAS with
  nonterminals $X_0, \ldots, X_m$ and with production rules
  $X_0 \pstep 1$ and $X_i \pstep -1 \: X_i \: X_{i-1} \mid 1 X_{i-1}$ for $1\le i\le m$.
  It is routinely checked that $\max\{d \mid c \vstep{X_m} d\}=A_m(c)$ for all $c\in\setN$.
  \qed
\end{example}

Every GVAS can be effectively normalized, in logarithmic space, by
replacing terminals $a \in \setZ$ by words over the alphabet $\{-1,0,1\}$
and then putting the resulting grammar into Chomsky normal form.
In addition,
non-productive nonterminals, and production rules in which they occur,
can be removed. %
So in order to simplify our proofs,
we consider w.l.o.g.~only GVAS of this simpler form.

\begin{assumption}
    We restrict our attention to GVAS $G=(V,A,R,S,c_\init)$
  in Chomsky normal form and where
  $A = \{-1,0,1\}$ and every $X \in V$ is productive.
\end{assumption}

The rest of the paper is devoted to the proof of \cref{thm:main}.
Before delving into its technical details,
we give a high-level description the proof.
In the next section,
we characterize unboundedness in terms of certificates,
which are complete parse trees whose nodes are labeled by natural numbers
(or $-\infty$).
These certificates contain a growing pattern that can be pumped to produce
infinitely many reachable ($1$-dimensional) vectors,
thereby witnessing unboundedness.
We then prove that certificates need not be too large.
To do so,
we first show in \cref{sec:struct} how to bound the size of growing patterns.
Then,
we bound the height and labels of ``minimal'' certificates
in \cref{sec:small_certs}.
Both bounds are singly-exponential in the size of the GVAS.
Thus,
the existence of a certificate can be checked by an alternating Turing machine
running in polynomial space.
This entails the desired {\EXPTIME} upper-bound stated in \cref{thm:main}.

\section{Certificates of Unboundedness}
\label{sec:certificates}
Following our previous work on the GVAS coverability problem~\cite{LerouxSutreTotzke:2015:ICALP},
we annotate parse trees
in a way that is consistent with the VAS semantics.
A \emph{flow tree} for a GVAS $G = (V,A,R,S,c_\init)$ is a complete\footnote{%
  Compared to~\cite{LerouxSutreTotzke:2015:ICALP} where
  flow trees are built on arbitrary parse trees,
  the flow trees that we consider here are always built on complete parse trees.
} parse tree $(T, \lsymoperator)$ for $G$
equipped with two functions $\linoperator, \loutoperator: T \to \setN \cup \{-\infty\}$,
assigning an \emph{input} and an \emph{output} value to each node,
with $\lin{\varepsilon} = c_\init$,
and satisfying, for every node $t \in T$,
the following \emph{flow conditions}:
\begin{enumerate}
\item
  \label{flow-conditions}
  If $t$ is internal with children $t0, \ldots, tk$, then
  $\lin{t0} \leq \lin{t}$,
  $\lout{t} \leq \lout{tk}$, and
  $\lin{t(j+1)} \leq \lout{tj}$ for every $j = 0, \ldots, k-1$.
\item
  If $t$ is a leaf, then
  $\lout{t} \leq \lin{t} + a$ if $\lsym{t} = a \in A$,
  and
  $\lout{t} \leq \lin{t}$ if $\lsym{t} = \varepsilon$.
\end{enumerate}
We shortly write $\lnode{t}{c}{\#}{d}$ to mean that
$(\lin{t}, \lsym{t}, \lout{t}) = (c, \#, d)$.
The \emph{size} of a flow tree is the size of its underlying parse tree.
\Cref{fig:certificates} (left) shows a flow tree for the GVAS of \cref{ex:gvas},
with start symbol $X_1$ and initial ($1$-dimensional) vector $c_\init = 5$.

\begin{remark}
  The flow conditions enforce the VAS semantics
  along a depth-first pre-order traversal of the complete parse tree.
  But,
  as in~\cite{LerouxSutreTotzke:2015:ICALP},
  we only require inequalities instead of equalities.
  This corresponds to a \emph{lossy} VAS semantics,
  where the counter can be non-deterministically decreased~\cite{DBLP:conf/stacs/BouajjaniM99}.
  The use of inequalities in our flow conditions simplifies the presentation
  and allows for certificates of unboundedness with smaller input/output values.
  Note that equalities would be required to get certificates of reachability,
  but the latter problem is out of the scope of this paper.
\end{remark}

\begin{restatable}{lemma}{LemExistenceFT}\label{lem:existence-of-flow-trees}
  For all $d$ with $c_\init \vstep{S} d$,
  there exists a flow tree with $\lout{\varepsilon} = d$.
\end{restatable}

Our main ingredient to prove \cref{thm:main} is a small model property.
First,
we show in this section that unboundedness can always be witnessed
by a flow tree of a particular form,
called a certificate (see \cref{def:certificate} and \cref{fig:certificates}).
Then,
we will provide in \cref{thm:small_certificates} exponential bounds on the height and
input/output values of ``minimal'' certificates.
This will entail the desired {\EXPTIME} upper-bound for
the prefix-closed GVAS boundedness problem.

We start by bounding the size of flow trees that do not contain
an iterable pattern, i.e., a nonterminal that repeats, below it,
with a larger or equal input value.
Formally,
a flow tree $(T, \lsymoperator, \linoperator, \loutoperator)$ is
called \emph{good} if
it contains a node $t$ and a proper ancestor $s \pprefix t$ such that
$\lsym{s} = \lsym{t}$
and
$\lin{s} \leq \lin{t}$.
It is called \emph{bad} otherwise.
We bound the size of bad flow trees by
(a) translating them into bad nested sequences, and
(b) using a bound given in~\cite{LPS2014} on the length of bad nested sequences.
Let us first recall some notions and results from~\cite{LPS2014}.
Our presentation is deliberately simplified and limited to our setting.

Let $(S, \preceq, \norm{\cdot})$ be the normed quasi-ordered set defined by
$S \eqdef V \times \setN$,
$(X, m) \preceq (Y, n) \equivdef X = Y \wedge m \leq n$, and
$\norm{(X, m)} = m$.
A \emph{nested sequence} is a finite sequence
$(s_1,h_1), \ldots, (s_\ell,h_\ell)$ of elements in $S \times \setN$
satisfying $h_1 = 0$ and $h_{j+1} \in h_j + \{-1,0,1\}$ for every index
$j < \ell$ of the sequence.
A nested sequence $(s_1,h_1), \ldots, (s_\ell,h_\ell)$ is called \emph{good}
if there exists $i < j$ such that $s_i \preceq s_j$
and $h_i \leq h_{i+1}, \ldots, h_j$.
A \emph{bad} nested sequence is one that is not good.
A nested sequence $(s_1,h_1), \ldots, (s_\ell,h_\ell)$
is called \emph{$n$-controlled}, where $n \in \setN$,
if $\norm{s_j} < n+j$ for every index $j$ of the sequence.

\begin{theorem}[{\cite[Theorem VI.1]{LPS2014}}]
  \label{thm:bad-nested-sequences}
  Let $n \in \setN$ with $n \geq 2$.
  Every $n$-controlled bad nested sequence has length at most
  $F_{\omega \ldotp \card{V}}(n)$.
\end{theorem}

The function $F_{\omega \ldotp \card{V}} : \setN \to \setN$ used in the
theorem is part of the \emph{fast-growing hierarchy}.
Its precise definition (see, e.g.,~\cite{LPS2014}) is not important
for the rest of the paper.
The following lemma provides a bound on the size of bad flow trees.
Notice that this lemma applies to arbitrary GVAS (not necessarily prefix-closed).

\begin{restatable}{lemma}{LemBoundBadFT}\label{lem:bound-on-bad-flow-trees}
  Every bad flow tree has at most $F_{\omega \ldotp \card{V}}(c_\init + 2)$ nodes.
\end{restatable}

A good flow tree contains an iterable pattern that can be ``pumped''.
However, the existence of such a pattern does not guarantee unboundedness.
For that,
we need stronger requirements on the input and output values,
as defined below.

\begin{definition}[Certificates]
  \label{def:certificate}
  A \emph{certificate} for a given GVAS is a flow tree
  $(T, \lsymoperator, \linoperator, \loutoperator)$ equipped with
  two nodes $s \pprefix t$ in $T$
  such that
  $$
  \lsym{s} = \lsym{t}
  \quad\text{and}\quad
  \lin{s} \leq \lin{t}
  \quad\text{and}\quad
  \lin{s} < \lin{t} \ \text{or} \ \lout{t} < \lout{s}
  $$
\end{definition}

We now present the main result of this section,
which shows that unboundedness can always be witnessed
by a certificate.

\begin{figure}[t]
  \tikzset{
    internal/.style={rectangle,draw=gray,thick},
    leaf/.style={rectangle,draw=gray,thick},
    iovalue/.style={inner sep=0pt,
            minimum size=0pt,
            font=\scriptsize,
            label distance=0cm,
        },
    lout/.style={label={[iovalue,text=blue]0:#1}},
    lin/.style={label={[iovalue, text=red]180:#1}},
    index/.style={label={[font=\scriptsize, label distance=0.6cm,]left:{#1:}}},
    every node/.style={inner sep=0pt, minimum size=5mm,text centered,},
}
\begin{minipage}[t]{.4\textwidth}
\centering
\begin{tikzpicture}[
    index/.style={},
    level distance=1.6cm,
]
\node [internal, index=$\eps$, lin=$5$, lout=$-\infty$] at (0,0){$X_1$}
  child {node[leaf,index=0, lin=5,lout=4] (0) {$-1$}}
  child {
      node[internal, index=1, lin=4,lout=5] (1) {$X_1$}
      child { node[leaf, index=10,lin=3,lout=4] (10) {$1$} }
      child {
          node[internal, index=11, lin=4,lout=5] (11) {$X_0$}
          child {
              node[leaf, index=110, lin=4, lout=5] (110) {$1$}
          }
      }
      child [missing]
  }
  child {
      node[internal, index=02, lin=5,lout=$-\infty$] (02) {$X_0$}
      child {
          node[leaf, index=020, lin=$-\infty$,lout=$-\infty$] (20) {$1$}
      }
  };
\end{tikzpicture}
\end{minipage}
\begin{minipage}[t]{.6\textwidth}
  \centering
\begin{tikzpicture}[
    auxnode/.style={inner sep=0pt, minimum height=0cm},
    MB/.style={-,decorate,decoration={snake, pre length=0cm}},
    node distance=1.2cm and 0.5cm,
    ]
    \path[use as bounding box] (-2.8,0) rectangle (2.8,-5);
    \node [internal, index=$\eps$, lin=$c_{\init}$, lout=$-\infty$] at (0,0) (ROOT) {${S}$};
    \node [internal,
           below= of ROOT,
           index=$s$,
           lin=$\lin{s}$,
           lout=$\lout{s}$] (s) {$X$};
    \node [internal,
           below= of s,
           index=$t$,
           lin=$\lin{t}$,
           lout=$\lout{t}$]  (t) {$X$};
    \draw [MB] (ROOT) -- (s);
    \draw [MB] (s) -- (t);

    \node [auxnode,below= of t] (M) {};
    \node [auxnode,left of=M, xshift=0.5cm] (LLL) {};
    \node [auxnode,left= of LLL] (LL) {};
    \node [auxnode,left= of LL](L) {};

    \node [auxnode,right of=M, xshift=0-.5cm] (RRR) {};
    \node [auxnode,right= of RRR] (RR) {};
    \node [auxnode,right= of RR](R) {};

    \path[draw] (ROOT) edge[out=south west, in=90] (L);
    \path[draw] (s)    edge[out=south west, in=90] (LL);
    \path[draw] (t)    edge[out=south west, in=90] (LLL);

    \path (L) 
      -- node {$x$} (LL)
      -- node {$u$} (LLL)
      -- node {$w$} (RRR)
      -- node {$v$} (RR)
      -- node {$y$} (R);
    
      \path[draw] (ROOT) edge[out=south east, in=90] (R);
      \path[draw] (s)    edge[out=south east, in=90] (RR);
      \path[draw] (t)    edge[out=south east, in=90] (RRR);
\end{tikzpicture}
\end{minipage}
\caption{{\bf Left:} a flow tree for the GVAS of \cref{ex:gvas}
    with $c_{\init}=5$.
    Input and output values are indicated in red and blue, respectively.
    {\bf Right:} A certificate with $\lsym{t}=\lsym{s}=X$
    and yield $xuwvy\in A^*$.
    It must hold that either
    $\lin{s} < \lin{t}$ or
    $\lin{s} = \lin{t}$ and $ \lout{t} <\lout{s}$.
    \label{fig:certificates}}
\end{figure}
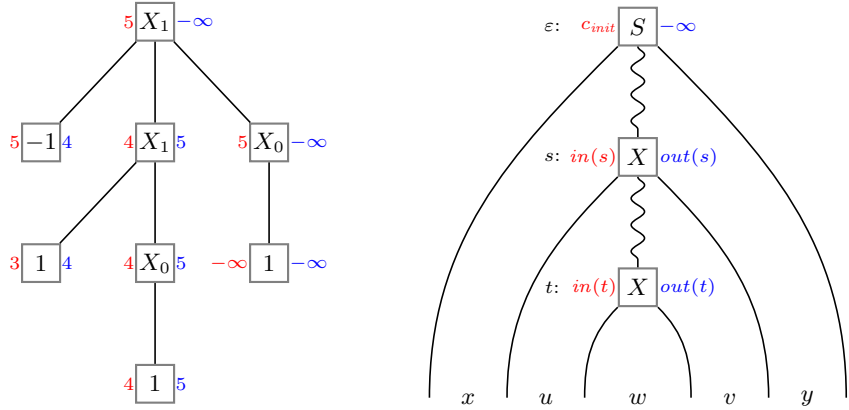

\begin{restatable}{theorem}{ThmCerts}
    \label{thm:certs}
  A prefix-closed GVAS $G$ is unbounded
  if, and only if,
  there exists a certificate for $G$.
\end{restatable}

\section{Growing Patterns}
\label{sec:struct}
Certificates depicted on \cref{fig:certificates} (right)
introduce words $u\in A^*$ satisfying a sign constraint $\sum u>0$ or
$\sum u=0$. These 
words are derivable from words of non-terminal symbols
$S_1\ldots S_k$ corresponding to the left children of the nodes between $s$
and $t$. In order to obtain small certificates, in this section, we
provide bounds on the minimal length of words $u'\in A^*$ that can
also be derived from
$S_1\ldots S_k$ and that satisfy the same sign constraint
as $u$.

\smallskip
Let us first introduce the \emph{displacement} of a GVAS $G$ as the ``best shift'' achievable by a word in $\lang[^G]{}$ and defined by the following equality\footnote{Notice that $\displ[^G]{}$ may be negative.}:
$$
\displ[^G]{} \ \eqdef \ \sup\{\textstyle\sum z \mid z \in \lang[^G]{}\}
$$

When the displacement is finite, the following \cref{lem:elementary} shows that it is achievable by a complete elementary parse tree. We say that a parse tree $T$ is \emph{elementary} if for every $s\prefix t$ such that $\lsym{s}=\lsym{t}$, we have $s=t$. Notice that the size of an elementary parse tree is bounded by $2^{|V|+1}$.
\begin{restatable}{lemma}{LemElementary}
    \label{lem:elementary}
  Every GVAS $G$ admits a complete elementary parse with a yield $w$ such that $\displ[^G]{}\in\{\sum w,+\infty\}$.
\end{restatable}

Given a non-terminal symbol $X$, we denote by $G[X]$ the context-free
grammar obtained from $G$ by replacing the start symbol by $X$.
We are now ready to state the main observation of this section.
\begin{theorem}\label{thm:derivewitness}
  For every sequence $S_1,\ldots,S_k$ of non-terminal symbols of a GVAS $G$ there exists a sequence $T_1,\ldots,T_k$ of complete parse trees $T_j$ for $G_j\eqdef G[S_j]$ with a yield $z_j$ such that $|T_1|+\cdots+|T_k|\leq 3k4^{|V|+1}$, and such that $\sum z_1\ldots z_k>0$ if $\displ[^{G_1}]{}+\cdots+\displ[^{G_k}]{}>0$, and $\sum z_1\ldots z_k=0$ if $\displ[^{G_1}]{}+\cdots+\displ[^{G_k}]{}=0$.
\end{theorem}

We first provide bounds on complete parse trees that witness the following properties $\displ[^G]{}=+\infty$ and $X$ is derivable. Formally, a nonterminal $X$ is said to be \emph{derivable} if there exists $w\in (A\cup V)^*$ that contains $X$ and such that $S\gstep{*}w$.
\begin{restatable}{lemma}{LemDisplInfinite}
    \label{lem:displinfinite}
  If $\displ[^G]{}=+\infty$, there exists a parse tree for $G[X]$ where $X$ is a non-terminal symbol derivable from the start symbol $S$ with a yield $uXv$ satisfying $u,v\in A^*$, $\sum uv>0$, and a number of nodes bounded by $4^{|V|+1}$. 
\end{restatable}

\begin{restatable}{lemma}{LemDerivable}
\label{lem:derivable}
  For every derivable non-terminal symbol $X$, there exists a parse tree with a yield in $A^* X A^*$ and a number of nodes bounded by $4^{|V|+1}$.
\end{restatable}

\begin{proof}[of \cref{thm:derivewitness}]
We can assume that $k\geq 1$ since otherwise the proof is trivial. Observe that
if $\displ[^{G_1}]{}+\cdots+\displ[^{G_k}]{}<+\infty$ then
$\displ[^{G_j}]{}<+\infty$ for every $j$. It follows from \cref{lem:elementary}
that there exists a complete parse tree $T_j$ for $G[S_j]$ with a yield $w_j$
satisfying $\displ[^{G_j}]{}=\sum w_j$ and a number of nodes bounded by
$2^{|V|+1}$. Thus $|T_1|+\cdots+|T_k|\leq k2^{|V|+1}$ and $\sum w_1\ldots
w_k=\displ[^{G_1}]{}+\cdots+\displ[^{G_k}]{}$.  So, in this special case the
theorem is proved. Now, let us assume that
$\displ[^{G_1}]{}+\cdots+\displ[^{G_k}]{}=+\infty$.
There exists $p\in\{1,\ldots,k\}$ such that $\displ[^{G_p}]{}=+\infty$.
\Cref{lem:displinfinite} shows that there exists a variable for $X$ derivable
from $S_p$ and a parse tree $T_+$ for $G[X]$ with a yield $uXv$ satisfying
$u,v\in A^*$, $\sum uv>0$, and such that $|T_+|\leq 4^{|V|+1}$. Since $S_j$ is
productive, there exists a complete elementary parse tree $T_j$ for $G[S_j]$
with a yield $w_j\in A^*$. For the same reason, there exists a complete
elementary parse tree $T$ for $G[X]$ with a yield $w\in A^*$. As $X$ is derivable
from $S$, \cref{lem:derivable} shows that there exists a parse tree $T'$ for
$G$ with a yield labeled by a word in $u' X v'$ with $u',v'\in A^*$, and a
number of nodes bounded by $4^{|V|+1}$. Notice that for any $n\in\setN$, we
deduce a complete parse tree $T_p$ for $G[S_p]$ with a yield $w_p=u'u^nw v^nv'$
by inserting in $T'$ many ($n$) copies of $T_+$ and one copy of $T$. Observe
that $\sum w_1\ldots w_k\geq -|w_1\ldots w_{p-1}w_{p+1}\ldots
w_k|-|u'wv'|+n\sum uv\geq -k2^{|V|+1}-4^{|V|+1}+n$. Let us fix $n$ to
$2k4^{|V|+1}-2$. It follows that $\sum w_1\ldots w_p>0$. Moreover, we have
$|T_p|\leq |T|-1+n(|T_+|-1)+|T'|\leq 2^{|V|+1}+n4^{|V|+1}+ 4^{|V|+1} \leq (n+2)4^{|V|+1}\leq 2k4^{|V|+1}$.
We derive $|T_1|+\cdots+|T_k|\leq (k-1)2^{|V|+1}+ 2k4^{|V|+1}\leq 3k4^{|V|+1}$.
We have proved \cref{thm:derivewitness}. \qed
\end{proof}

\section{Small Certificates}
\label{sec:small_certs}
We provide in this section exponential bounds on the height and
input/output values of minimal certificates
in the following sense.
Let the \emph{rank} of a flow tree
$(T, \lsymoperator, \linoperator, \loutoperator)$
be the pair
$$
\left(
  \card{T_{\linoperator}} \,+\, \card{T_{\loutoperator}}
  \ ,\ 
  \sum_{t \in T_{\linoperator}} \lin{t} \,+\, \sum_{t \in T_{\loutoperator}} \lout{t})
\right)
$$
where
$T_{\linoperator} = \{t \in T \mid \lin{t} > -\infty\}$
and
$T_{\loutoperator} = \{t \in T \mid \lout{t} > -\infty\}$.
Notice that $T_{\loutoperator} \subseteq T_{\linoperator}$.
We compare ranks using the 
lexicographic order $\lex$ over $\setN^2$
and let the rank of a certificate $(\mathcal{T}, s, t)$ be
the rank of its flow tree $\mathcal{T}$.

\medskip
Consider a prefix-closed GVAS $G = (V,A,R,S,c_\init)$ that is unbounded.
By \cref{thm:certs}, there exists a certificate for $G$.
Pick a certificate
$(\mathcal{T}, s, t)$
among those of least rank.
Our goal is to bound the height and input/output values of $\mathcal{T}$.
Based on its assumed minimality, we observe a series of facts about
our chosen certificate.

\medskip
First, we observe that some input/output values in $\mathcal{T}$
must be $-\infty$, because higher values would be useless
in the sense that they can be set to $-\infty$ without
breaking the flow conditions nor the conditions on $s$ and $t$.
This observation is formalized in the two following facts.

\begin{fact}
  \label{fact:dropped-values-1}
  It holds that $\lout{p} = -\infty$ for every proper ancestor $p \pprefix s$.
  Moreover,
  $\lin{p} = \lout{p} = -\infty$ for every node $p \in T$ such that
  $s \lexst p$ and $p \not\prefix s$.
\end{fact}

\begin{fact}
  \label{fact:dropped-values-2}
  Assume that $\lin{s} < \lin{t}$.
  It holds that $\lout{p} = -\infty$ for every ancestor $p \prefix t$.
  Moreover,
  $\lin{p} = \lout{p} = -\infty$ for all $p \in T$ with $t \lexst p$.
\end{fact}

Next, we observe that the main branch, that contains $s$ and $t$,
must be short.

\begin{restatable}{fact}{FactBranchDepthBound}
  \label{fact:branch-depth-bound}
  It holds that $\len{s} \leq \card{V}$ and $\len{t} \leq \len{s} + \card{V} + 1$.
\end{restatable}

The next two facts provide \emph{relative} bounds on input and output values
for nodes that are not on the main branch.

\begin{restatable}{fact}{FactTraversalInequation}
  \label{fact:traversal-inequation}
  It holds that $\lin{p} \leq \lout{p} + 2^{\card{V}}$
  for every node $p \in T$ with $p \not\prefix t$.
\end{restatable}

\begin{restatable}{fact}{FactDescentOutsideMB}
  \label{fact:descent-outside-main-branch}
  Let $q \in T$ and let $p$ be the parent of $q$.
  If $p = t$ or $p \not\prefix t$,
  then $\lout{q} \leq \lout{p} + 2^{\card{V}}$.
  If moreover $\lsym{p} = \lsym{q}$,
  then $\lout{q} < \lout{p}$.
\end{restatable}

The following facts provide \emph{absolute} bounds
on the input/output values of nodes $s$ and $t$.
The proofs of the facts below %
crucially rely on \cref{sec:struct}.
Consider the subtrees on the left and on the right of the branch from $s$ to $t$.
The main idea of the proofs is to replace these subtrees by small ones
using \cref{thm:derivewitness}.

\begin{restatable}{fact}{FactBranchValuesBoundOne}
  \label{fact:branch-values-bound-1}
  It holds that $\lout{t} \leq \lout{s} \leq 6 \card{V} \cdot 4^{\card{V}+1}$.
\end{restatable}

\begin{restatable}{fact}{FactBranchValuesBoundThree}
  \label{fact:branch-values-bound-3}
  It holds that $\lin{s} \leq \lin{t} \leq 7 \card{V} \cdot 4^{\card{V}+1}$.
\end{restatable}

Now we derive absolute bounds for the input/output values
of the remaining nodes on the main branch.
These are derived from
\cref{fact:branch-values-bound-1,fact:branch-values-bound-3},
using \cref{fact:traversal-inequation,fact:descent-outside-main-branch}
about the way in/output values propagate and the \cref{fact:branch-depth-bound}
that the intermediate path between nodes $s$ and $t$ is short.

\begin{restatable}{fact}{FactBranchValuesBoundTwo}
  \label{fact:branch-values-bound-2}
  It holds that $\lout{p} \leq 4^{2(\card{V}+1)}$
  for every ancestor $p \prefix t$.
\end{restatable}

\begin{restatable}{fact}{FactBranchValuesBoundFour}
  \label{fact:branch-values-bound-4}
  It holds that $\lin{p} \leq 4^{2(\card{V}+1)}$
  for every ancestor $\varepsilon \pprefix p \prefix t$.
\end{restatable}

We are now ready to derive bounds on the rank of our minimal certificate.
Notice that it remains only to bound the depth and the input/output values
on branches different from the main branch.

Consider therefore a node $q$ outside the main branch,
i.e.,
$q \not\prefix t$.
Let $p$ be the least prefix of $q$ such that $p = t$ or $p \not\prefix t$.
We first show that 
$\lout{p} \leq 4^{2(\card{V}+1)}$.
If $p = t$ then the claim follows from \cref{fact:branch-values-bound-2}.
Otherwise,
the parent $r$ of $p$ satisfies $r \pprefix t$.
Observe that the other child $\bar{p}$ of $r$
satisfies $\bar{p} \prefix t$.
The flow conditions together with
the minimality of $(\mathcal{T}, s, t)$ guarantee that
\begin{itemize}
\item
  if $p = r1$ then $\lout{p} = \lout{r}$,
  hence, $\lout{p} \leq 4^{2(\card{V}+1)}$
  by \cref{fact:branch-values-bound-2},
  and
\item
  if $p = r0$ then $\lout{p} = \lin{r1}$,
  hence, $\lout{p} \leq 4^{2(\card{V}+1)}$
  by \cref{fact:branch-values-bound-4}.
\end{itemize}
According to \cref{fact:descent-outside-main-branch},
the output values on the branch from $p$ down to $q$ may only increase
when visiting a new symbol.
Moreover, this increase is bounded by $2^{\card{V}}$.
It follows that $\lout{r} \leq \lout{p} + \card{V} 2^{\card{V}}$
for every node $r$ such that $p \pprefix r \prefix q$.
\cref{fact:traversal-inequation} entails that
$\lin{r} \leq \lout{p} + (\card{V} + 1) 2^{\card{V}}$.
We obtain that
$\max\{\lin{r}, \lout{r}\} < 4^{3(\card{V}+1)}$
for every node $r$ with $p \pprefix r \prefix q$.
\cref{fact:descent-outside-main-branch} also forbids
the same nonterminal from appearing twice with the same output value,
so
$\len{r} \leq \len{p} + \card{V}\cdot 4^{3(\card{V}+1)} + 1$.
Observe that $\len{p} \leq \len{t}$. %
We derive from \cref{fact:branch-depth-bound} that
$\len{r} \leq 4^{4(\card{V}+1)}$.
This concludes the proof of the following theorem.

\begin{theorem}
  \label{thm:small_certificates}
  A prefix-closed GVAS $(V,A,R,S,c_\init)$ is unbounded if, and only if,
  it admits a certificate with height and all input/output values
  bounded by $c_\init + 4^{4(\card{V}+1)}$.
\end{theorem}

\begin{proof}[of \cref{thm:main}]
    By \cref{thm:small_certificates}, a certificate for unboundedness
    is a flow tree of exponential height and with all input and output labels
    exponentially bounded.
    An alternating Turing machine
    can thus guess and verify all branches of such a flow tree,
    storing intermediate
    input/output values as well as the remaining length of a branch
    in polynomial space.
    The claim then follows from the fact that
    alternating polynomial space equals exponential time.
    \qed
\end{proof}

\section{Conclusion}
We discussed different boundedness problems for 
pushdown vector addition systems \cite{LPS2014,Laz2013},
which are a known, and very expressive computational model
that features nondeterminism, a pushdown stack and
several counters.
These systems may be equivalently interpreted, in the context of regulated
rewriting \cite{Das1997}, as vector addition systems
with context-free control languages.

We observe that boundedness is reducible to
both counter- and stack-boundedness.
The stack boundedness problem can be shown to be decidable
(with hyper-Ackermannian complexity) by adjusting
the algorithm presented in \cite{LPS2014}.

Here, we single out the special case of the counter-boundedness problem
for one-dimensional systems and propose an exponential-time
algorithm that solves it.
This also improves the best previously known
Ackermannian upper bound
for boundedness in dimension one.

Currently, the best lower bound for this problem is \NP, which can 
be seen by reduction from the subset sum problem.
For dimension two, \PSPACE-hardness
follows by reduction from the state-reachability of bounded one-counter automata
with succinct counter updates \cite{FJ2013}.
For arbitrary dimensions, \TOWER-hardness is known already for the
boundedness problem
\cite{Laz2013,LPS2014}
but the decidability of counter-boundedness for PVAS remains open.

\subsection*{Acknowledgments}
The authors wish to thank M. Praveen for insightful discussions.
We also thank the anonymous referees for their useful comments and suggestions.

\bibliographystyle{splncs03}
\bibliography{references}

\clearpage
\appendix
\section{Missing Proofs}
\subsection{Proofs for \cref{sec:pdvass}}
\label{appendix.pdvass}
\LemBoundednessReduction*
\begin{proof}
  Let us consider a PVAS
  $A$
  and let us introduce a PVAS $A'$ such that $A$ is bounded if, and only if, $A'$ is bounded,
  and such that if $A$ is unbounded then $A'$ is both counter-unbounded
  and stack-unbounded.
  The system $A'$ is a copy of $A$, extended with a new state $\perp$.
  The state $\perp$ has self-loops that allow to pop any symbol from the stack and simultaneously increment
  the first counter. It also has self-loops that decrease any counter and push
  some symbol to the stack.
  Finally, we add a transition
  $(q,\vec{0},\nop,\perp)$
  for each original state $q$.
  Now just observe that $A$ is bounded if, and only if, $A'$ is
  bounded. Moreover if $A$ is unbounded, then $A'$ is both
  counter-unbounded and stack-unbounded.
    \qed
\end{proof}

\setcounter{subsection}{2}
\subsection{Proofs for \cref{sec:model}}
\LemPVAStoGVAS*
\begin{proof}
    Just observe that a PVAS can be interpreted as a pushdown automaton
    that recognizes a context-free and prefix-closed trace language
  $L\subseteq \vec{A}^*$ where $\vec{A}\subseteq\Z$ is the set of vectors
  labeling the transitions. %
  We can
  construct, in logarithmic space, a context-free grammar that produces $L$. This
  context-free grammar (equipped with the initial value $\vec{c}_\init$ of the PVAS)
  is a prefix-closed GVAS $G$. The
  reachability set of $G$ is exactly the set of vectors $\vec{c}$ such
  that $(q,\vec{c},w)$ is a reachable configuration of the PVAS for some $q$ and $w$. The converse
  construction follows a similar idea by observing that the language
  of a prefix-closed context-free grammar can be accepted by a pushdown
  automaton (with all states accepting), computable in logarithmic space.
    \qed
\end{proof}

\label{appendix:normalization}
\begin{lemma}
\label{lem:normalization}
    Let $G=(V,A,R,S,c_{\init})$
    be a GVAS with
    $\max\{|a|: a\in A\} \le n$.
    One can construct, in logspace, an equivalent
    GVAS
    $G'=(V',A',R',S',c_{\init})$
    with the same reachability set
    and such that
    $A'=\{-1,0,1\}$.
\end{lemma}
\begin{proof}
    $G'$ will be a copy of $G$, extended as follows.
    For all $1\le m\le n$, there are new nonterminals $B_m$ and rules
    $B_1\pstep 1$ and
    $B_m\pstep B_{m-1}B_{m-1}$
    for $m>0$.
    Now all terminals $a\in \vec{A}$ such that $a>1$
    are removed from $A'$ and replaced (on right hand sides of all rules)
    by a new nonterminal $X_a$.
    The only rule that rewrites this symbol is
    \begin{equation}
    X_a\pstep B_n^{b_n}B_{n-1}^{b_{n-1}}\cdots B_1^{b_1}
    \end{equation}
    where $a=b_{n}b_{n-1}\dots b_1$ is the binary representation of $a$.
    Note that $B_{m}^{b_{m}}$ is the empty word if $b_m=0$ and $B_{m}$
    otherwise.
    We thus observe that
    the language $\lang[^{G'}]{X_a}$ contains only
    the word $1^{a}$. In particular,
    for any $c,d\in\N$ we get that
    $c\vstep{a}d$ iff $c\vstep{X_a}d$.

    An analogous construction allows to replace all
    terminals $a\in \vec{A}$ with $a<1$.
    The resulting GVAS $G'$ has terminal alphabet $A'=\{-1,0,1\}$
    as required.
\qed
\end{proof}

\subsection{Proofs for \cref{sec:certificates}}
\LemExistenceFT*
\begin{proof}
  Assume that $c_\init \vstep{S} d$.
  It holds that $c_\init \vstep{z} d$ for some $z \in \lang{S}$.
  Since $z \in \lang{S}$,
  there exists a derivation $S \gstep{*} z$, hence,
  a complete parse tree with root labeled by $S$ and yield $z$.
  This complete parse tree,
  together with the fact that $c_\init \vstep{z} d$,
  induces a flow tree with root $\lnode{\varepsilon}{c_\init}{S}{d}$.
  \qed
\end{proof}

\LemBoundBadFT*
\begin{proof}
  Let $\mathcal{T} = (T, \lsymoperator, \linoperator, \loutoperator)$
  be a flow tree.
  We construct a nested sequence that corresponds to
  a depth-first pre-order traversal of the flow tree.
  Let us introduce, for each symbol $X \in V$,
  two copies $X'$ and $X''$.
  We associate to each node $t \in \setN^*$ a word $\theta(t)$
  over $S \times \setN$,
  inductively defined as follows:
  $$
  \theta(t) \ = \ 
  \begin{cases}
    ((X, m),   \len{t}) \cdot \theta(t0) \cdot
    ((X', m),  \len{t}) \cdot \theta(t1) \cdot
    ((X'', m), \len{t})
    & \text{if} \ t0 \in T\\
    \varepsilon
    & \text{otherwise}
  \end{cases}
  $$
  where $X = \lsym{t}$ and $m = \lin{t}$.
  Recall that the condition $t0 \in T$ means that $t$ is an internal node of $T$.
  It is readily seen that $\theta(\varepsilon)$ is a nested sequence.
  Let us write it as $\theta(\varepsilon) = (s_1,h_1), \ldots, (s_\ell,h_\ell)$.
  Obviously,
  every index $1 \leq i \leq \ell$ can be mapped back to a node $t(i)$ of the
  flow tree $\mathcal{T}$.

  Assume that $\mathcal{T}$ is bad,
  and suppose, towards a contradiction, that $\theta(\varepsilon)$ is good.
  So there exists $i < j$ such that $s_i \preceq s_j$
  and $h_i \leq h_{i+1}, \ldots, h_j$.
  This entails that $t(i)$ is an ancestor of $t(j)$.
  Moreover, $t(i) \neq t(j)$ because each node $t \in T$ is visited three times
  in the sequence $\theta(\varepsilon)$,
  and each visit uses a different copy of $\lsym{t}$.
  So $t(i)$ is a proper ancestor of $t(j)$.
  Since $s_i \preceq s_j$,
  we get that $\lsym{t(i)} = \lsym{t(j)}$
  and
  $\lin{t(i)} \leq \lin{t(j)}$,
  which contradicts our assumption that $\mathcal{T}$ is bad.

  We have shown that the nested sequence $\theta(\varepsilon)$ is bad.
  Let us show that $\theta(\varepsilon)$ is $c_\init$-controlled.
  Let $1 \leq j \leq \ell$.
  Since $A = \{-1,0,1\}$ by assumption,
  it holds that $\lin{t(j)} \leq c_\init + L$ where $L$ denotes the number of
  leaves that are lexicographically smaller than $t(j)$.
  Recall that $G$ is in Chomsky normal form by assumption.
  So these leaves have distinct parents,
  and those are all visited before $t(j)$ in the sequence
  $\theta(\varepsilon)$, hence, $j > L$.
  It follows that $\norm{s_j} = \lin{t(j)} < c_\init + j$.

  Since the nested sequence $\theta(\varepsilon)$ is $(c_\init + 2)$-controlled
  and bad,
  we derive from~\cref{thm:bad-nested-sequences} that
  its length $\ell$ satisfies $\ell \leq F_{\omega \ldotp \card{V}}(c_\init + 2)$.
  The observation that $\card{T} \leq \ell$ concludes the proof.
  \qed
\end{proof}

\ThmCerts*
\begin{proof}
  Assume that the reachability set of $G$ is infinite.
  So there exists $d$ such that
  $c_\init \vstep{S} d$ and $d > c_\init + F_{\omega \ldotp \card{V}}(c_\init+2)$.
  Pick a flow tree
  $\mathcal{T} = (T, \lsymoperator, \linoperator, \loutoperator)$
  with root $\lnode{\varepsilon}{c_\init}{S}{d}$,
  among those of least size.
  Note that such a flow tree exists by \cref{lem:existence-of-flow-trees}.
  Let $z \in A^*$ denote the yield of the complete parse tree
  $(T, \lsymoperator)$.
  It is readily seen that $c_\init \vstep{z} e$ for some $e \geq d$.
\gsnote{State a lemma for that?}
  Recall that $A = \{-1,0,1\}$ by assumption.
  It follows that
  $\len{z} \geq \sum z = e-c_\init > F_{\omega \ldotp \card{V}}(c_\init+2)$.
  Observe that $\card{T} \geq \len{z}$
  since $z$ is the yield of $T$.
  It follows from \cref{lem:bound-on-bad-flow-trees} that
  the flow tree $(T, \lsymoperator, \linoperator, \loutoperator)$ is good.
  So it contains a node $t$ and a proper ancestor $s \pprefix t$ such that
  $\lsym{s} = \lsym{t}$
  and
  $\lin{s} \leq \lin{t}$.
  To prove that $(\mathcal{T},s,t)$ is a certificate for $G$,
  it suffices to show that $\lin{s} < \lin{t}$ or $\lout{t} < \lout{s}$.
  Assume, by contradiction,
  that $\lin{s} = \lin{t}$ and $\lout{t} \geq \lout{s}$.
  We may replace,
  without breaking the flow conditions,
  the subtree rooted in $s$ by the subtree rooted in $t$.
  We may even preserve the input and output of $s$.
  The resulting flow tree also has root $\lnode{\varepsilon}{c_\init}{S}{d}$,
  but it has less nodes than $\mathcal{T}$,
  which contradicts the minimality of $\mathcal{T}$.

  \smallskip

  Conversely,
  assume that there exists a certificate $(\mathcal{T},s,t)$ for $G$,
  with
  $\mathcal{T} = (T, \lsymoperator, \linoperator, \loutoperator)$.
  By definition,
  it holds that
  $s \pprefix t$,
  $
  \lsym{s} = \lsym{t}
  $,
  and either $\lin{s} < \lin{t}$ or
  $\lin{s} = \lin{t}$ and $\lout{t} < \lout{s}$.
  Let $X$ denote the common nonterminal $X = \lsym{s} = \lsym{t}$.
  We decompose the yield $z \in A^*$ of the complete parse tree
  $(T, \lsymoperator)$ into $z = x u w v y$,
  as depicted in \cref{fig:certificates},
  where:
  \begin{itemize}
  \item
    $x$ and $y$ come from the leaves that are lexicographically smaller
    and larger than $s$, respectively,
  \item
    $u$ and $v$ come from the leaves of the subtree rooted in $s$
    that are lexicographically smaller and larger than $t$, respectively,
    and
  \item
    $w$ comes from the leaves of the subtree rooted in $t$.
  \end{itemize}
  It is readily seen that $S \gstep{*} x X y$, $X \gstep{*} u X v$ and $X \gstep{*} w$.
  Since $\lang{S}$ is prefix-closed,
  we get that
  $\{x u^n \mid n \in \setN\} \subseteq \lang{S}$
  and
  $\{x u^n w v^n \mid n \in \setN\} \subseteq \lang{S}$.
  Observe that
  $\lin{t} \leq \lin{s} + \sum u$ and $\lout{s} \leq \lout{t} + \sum v$.
  These two inequalities follow from the flow conditions.
  There are two cases.
\gsnote{The remainder needs more explanations, ideally a lemma to refer to.}
  \begin{enumerate}
  \item
    Either $\lin{s} < \lin{t}$, in which case $\sum u > 0$.
    Since $-\infty < \lin{t}$,
    it holds that $c_\init \vstep{xu} d$ for some $d$.
    We derive that $c_\init \vstep{x} c \vstep{u^n} c + n \cdot \sum u$
    for every $n \in \setN$,
    where $c = c_\init + \sum x$.
    Since $\sum u > 0$,
    we derive that $\{d \mid c_\init \vstep{S} d\}$ is infinite.
  \item
    Or $\lin{s} = \lin{t}$ and $\lout{t} < \lout{s}$,
    in which case $\sum u \geq 0$ and $\sum v > 0$.
    Since $-\infty < \lout{s}$,
    it holds that $c_\init \vstep{xuwv} d$ for some $d$.
    We derive that $c_\init \vstep{xu^nw} c + n \cdot \sum u \vstep{v^n} c + n \cdot \sum uv$
    for every $n \in \setN$,
    where $c = c_\init + \sum xw$.
    Since $\sum uv > 0$,
    we derive that $\{d \mid c_\init \vstep{S} d\}$ is infinite.
  \end{enumerate}
  In both cases, we obtain that the reachability set of $G$ is infinite.
  \qed
\end{proof}

\subsection{Proofs for \cref{sec:struct}}
\LemElementary*
\begin{proof}
  Notice that if $\displ[^G]{}<+\infty$ there exists a complete parse tree with
  a yield $w$ such that $\sum w=\displ[^G]{}$. Since $G$ admits a complete parse
  tree, we can pick
  a complete parse tree $T$ that is minimal wrt.~the number of nodes
  and with a yield $w$ that satisfies
  $\displ[^G]{}\in\{\sum w,+\infty\}$.
  Assume by contradiction that $T$ is not elementary. In that
  case, there exist two distinct nodes $s\pprefix t$ that are labeled by the
  same non-terminal symbol $X\in V$. Notice that if $\displ[^G]{}=+\infty$, by
  collapsing in $T$ the nodes $s$ and $t$, we get a parse tree $T'$ with a yield
  $w'$ that naturally satisfies $\displ[^G]{}\in\{\sum w',+\infty\}$. Thus $T'$
  contradicts the minimality of $T$. We deduce that $\displ[^G]{}<+\infty$ and
  in particular $\sum w=\displ[^G]{}$.
  Let us decompose the yield $w$ as $w=aubvc$, where the subwords $u,v$
  derive from the pumpable path from node $s$ to $t$.
  If $\sum u+\sum v>0$, by inserting many
  copies of this subtree in $T$, we get $\displ[^G]{}=+\infty$, which is
  impossible. It follows that $\sum u+\sum v\leq 0$. By collapsing the nodes $s$
  and $t$, we get a complete parse tree $T'$ for $G$ such that $|T'|<|T|$ with a
  yield $w'$ satisfying $\sum w'=\sum w-(\sum u+\sum v)$.
  From $\sum u+\sum v\leq 0$ we derive $\sum w'\geq \sum w$. As $\sum w=\displ[^G]{}$ and
  $\sum w'\leq \displ[^G]{}$, we derive $\displ[^G]{}=\sum w'$ and we get a
  contradiction on the minimality of $T$. It follows that $T$ is elementary.
  \qed
\end{proof}

\LemDisplInfinite*
\begin{proof}
  Let us first prove that there exists a non-terminal symbol $X$ derivable from the start symbol $S$ and a parse tree for $G[X]$ with a yield $uXv$ satisfying $u,v\in A^*$ and $\sum uv>0$. Since $\displ[^G]{}=+\infty$, there exists a minimal (for the number of nodes) complete parse tree $T$ with a yield $w$ satisfying $\sum w>2^{|V|}$. Observe that if $T$ is elementary then $|w|\leq 2^{|V|}$ and in particular $\sum w\leq 2^{|V|}$ and we get a contradiction. So the tree $T$ is not elementary. Hence there exists $s\pprefix t$ in $T$ with $\lsym{s} = X= \lsym{t}$ for some non-terminal symbol $X$. The subtree of $T$ between $s$ and $t$ provides a parse tree for $G[X]$ with a yield $uXv$ where $u,v\in A^*$. If $\sum uv\leq 0$ by collapsing nodes $s$ and $t$ in $T$, we derive a complete parse tree $T'$ such that $|T'|<|T|$ with a yield $w'$ satisfying $\sum w'+\sum uv=\sum w$. Thus $\sum w'\geq 2^{|V|}$ and we get a contradiction on  the minimality of $|T|$. Thus $\sum uv>0$.

  In the previous paragraph, we have proved that there exists a non-terminal symbol $X$ derivable from the start symbol $S$ and a parse tree $T$ for $G[X]$ with a yield $uXv$ satisfying $u,v\in A^*$ and $\sum uv>0$. Without loss of generality, we can pick $X$ and $T$ is such a way $|T|$ is minimal. Let $t$ be the unique leaf of $T$ labeled by $X$ and assume by contradiction that $|t|>|V|$. In this case, there exists $r\pprefix s\pprefix t$ such that $\lsym{r}=X'=\lsym{s}$ for some non-terminal symbol $X'$. Notice that in that case the subtree between $r$ and $s$ is a parse tree $T'$ for $G[X']$ with a yield $u'X'v'$ where $u',v'\in A^*$. Since $|T'|<|T|$, by minimality of $T$, we get $\sum u'v'\leq 0$. In particular, by collapsing in $T$ the nodes $s$ and $t$ we get a parse tree $T''$ for $G[X]$ with a yield $u''Xv''$ such that $\sum u''v''+\sum u'v'=\sum uv$. Thus $\sum u''v''>0$. We get a contradiction on the minimality of $|T|$ since $|T''|<|T|$. Hence $|t|\leq |V|$. Symmetrically, observe that if there exists $r\pprefix s$ such that $\lsym{r}=X'=\lsym{s}$ for some non-terminal symbol $X'$ and such that $r\not\prefix t$ then we get a contradiction on the minimality of $T$. Therefore $T$ can be decomposed as a branch for the root to $t$ in such a way nodes on this branch, except $t$, have exactly, on the right or on the left an elementary subtree. Thus $|T|\leq |V|+1+|V|2^{|V|+1}$. Since $|V|+1\leq 2^{|V|+1}$, we get $|T|\leq 2^{|V|+1}+|V|2^{|V|+1}=(|V|+1)2^{|V|+1}\leq 4^{|V|+1}$.

  \qed
\end{proof}

\LemDerivable*
\begin{proof}
  Since $X$ is derivable from $S$, there exists a sequence $X_0,\ldots,X_k$ of non-terminal symbols with $X_0=S$, $X_k=X$, $k+1\leq |V|$, and a sequence of production rules $X_{j-1}\pstep \alpha_jX_j\beta_j$ with $\alpha_j\beta_j=Y_j$ for some non-terminal symbol $Y_j$. Notice that there exists a complete elementary parse tree $T_j$ for $G[Y_j]$. The parse trees $T_1,\ldots,T_k$ put along a branch labeled by $X_0,\ldots,X_k$ provide a parse tree for $G$ with a yield in $A^* X A^*$ and a number of nodes bounded by $(k+1)+k2^{|V|+1}\leq |V|+(|V|-1)2^{|V|+1}\leq |V|2^{|V|+1}\leq 4^{|V|+1}$.
  \qed
\end{proof}

\subsection{Proofs for \cref{sec:small_certs}}
\FactBranchDepthBound*
\begin{proof}
  Suppose, towards a contradiction, that $\len{s} > \card{V}$.
  There must exit two nodes $p \pprefix q \pprefix s$
  with $\lsym{p} = \lsym{q}$.
  If $\lin{p} < \lin{q}$ then we get a certificate $(\mathcal{T}', p, q)$
  of strictly smaller rank than $(\mathcal{T}, s, t)$ by setting the input value
  of $t$ to $-\infty$ and propagating onwards.
  If $\lin{p} \geq \lin{q}$ then we may replace
  the subtree rooted in $p$ by the subtree rooted in $q$,
  retaining the flow conditions since $\lout{p} = \lout{q} = -\infty$
  by \cref{fact:dropped-values-1},
  and thus get a certificate
  of strictly smaller rank than $(\mathcal{T}, s, t)$.
  Both cases contradict the minimality of $(\mathcal{T}, s, t)$.
  This concludes the proof that $\len{s} \leq \card{V}$.

  Now suppose, towards a contradiction, that $\len{t} > \len{s} + \card{V} + 1$.
  There exists necessarily two nodes $s \pprefix p \pprefix q \pprefix t$
  with $\lsym{p} = \lsym{q}$.
  If $\lin{p} < \lin{q}$ then we get a certificate $(\mathcal{T}', p, q)$
  of strictly smaller rank than $(\mathcal{T}, s, t)$ by setting the input value
  of $t$ to $-\infty$ and propagating onwards.
  Similarly,
  if $\lin{p} = \lin{q}$ and $\lout{q} < \lout{p}$
  then we get a certificate $(\mathcal{T}', p, q)$
  of strictly smaller rank than $(\mathcal{T}, s, t)$ by setting the output value
  of $s$ to $-\infty$ and propagating onwards.
  If $\lin{p} = \lin{q}$ and $\lout{q} \geq \lout{p}$ then we may replace
  the subtree rooted in $p$ by the subtree rooted in $q$,
  and thus get a certificate
  of strictly smaller rank than $(\mathcal{T}, s, t)$.
  The remaining case is when $\lin{p} > \lin{q}$.
  In that case,
  we collapse the nodes $p \pprefix q$,
  preserve the input value of $p$,
  and relabel all nodes lexicographically larger than $p$
  with the largest input/output values allowed by the flow conditions.
  In the resulting flow tree $\mathcal{T}'$,
  which has a smaller rank than $\mathcal{T}$,
  the node $t'$ originating from $t$ has a strictly larger input value than $s$.
  So $(\mathcal{T}', s, t')$ is a certificate.
  All cases contradict the minimality of $(\mathcal{T}, s, t)$.
  This concludes the proof that $\len{t} \leq \len{s} + \card{V} + 1$.
 \qed
\end{proof}

\FactTraversalInequation*
\begin{proof}
  Suppose, towards a contradiction,
  that $\lin{p} > \lout{p} + 2^{\card{V}}$ for some node $p \in T$
  with $p \not\prefix t$.
  Recall that $A = \{-1,0,1\}$ by assumption.
  If the subtree rooted in $p$ has at most $2^{\card{V}}$ leaves,
  then we may decrease the input and output values of its nodes,
  retaining the flow conditions,
  so that the output value of $p$ is preserved and its new input value
  is at most $\lout{p} + 2^{\card{V}}$.
  Notice that this does not modify the main branch since $p \not\prefix t$.
  Thus,
  we get a certificate of strictly smaller rank than $(\mathcal{T}, s, t)$.
  Otherwise,
  the subtree rooted in $p$ has at least $2^{\card{V}} + 1$ leaves.
  Hence,
  it is not elementary
  and we may reduce it into a strictly smaller,
  elementary subtree with at most $2^{\card{V}}$ leaves.
  The latter induces a complete flow tree
  with the same input and output values for $p$.
  Again,
  this does not modify the main branch since $p \not\prefix t$.
  Thus,
  we get a certificate of strictly smaller rank than $(\mathcal{T}, s, t)$.
  \qed
\end{proof}

\FactDescentOutsideMB*
\begin{proof}
  Assume that $p = t$ or $p \not\prefix t$.
  Observe that the children of $p$ are not on the main branch,
  i.e.,
  none of them is a prefix of $t$.
  If $q$ is the last child of $p$,
  then $\lout{q} = \lout{p}$ by minimality of $(\mathcal{T}, s, t)$.
  Otherwise,
  $q = p0$ and $p1$ is the last child of $p$.
  It holds that $\lout{p0} = \lin{p1}$ and $\lout{p1} = \lout{p}$
  by minimality of $(\mathcal{T}, s, t)$.
  We derive from \cref{fact:traversal-inequation} that
  $\lout{q} \leq \lout{p} + 2^{\card{V}}$.

  Now assume, in addition, that $\lsym{p} = \lsym{q}$.
  Suppose, towards a contradiction, that $\lout{q} \geq \lout{p}$.
  If $\lin{p} < \lin{q}$ then we get a certificate $(\mathcal{T}', p, q)$
  of strictly smaller rank than $(\mathcal{T}, s, t)$ by setting the input value
  of $t$ to $-\infty$ and propagating onwards.
  If $\lin{p} \geq \lin{q}$ then we may replace,
  retaining the flow conditions since $\lout{q} \geq \lout{p}$,
  the subtree rooted in $p$ by the subtree rooted in $q$,
  and thus get a certificate
  of strictly smaller rank than $(\mathcal{T}, s, t)$.
  Both cases contradict the minimality of $(\mathcal{T}, s, t)$.
  It follows that $\lout{q} < \lout{p}$,
  which concludes the proof of the fact.
  \qed
\end{proof}

\FactBranchValuesBoundOne*
\begin{proof}
  If $\lin{s} < \lin{t}$ then $\lout{s} = \lout{t} = -\infty$
  by \cref{fact:dropped-values-2},
  so the equality $\lout{s} = \lout{t} + 1$ holds.
  Otherwise, $\lout{t} < \lout{s}$.
  If we had $\lout{t} + 1 < \lout{s}$ then we could
  decrease the output value of $s$ by one,
  retaining the flow conditions by \cref{fact:dropped-values-1},
  and get a certificate
  of strictly smaller rank than $(\mathcal{T}, s, t)$,
  contradicting the minimality of $(\mathcal{T}, s, t)$.
  Therefore we get that
  \begin{equation}
  \label{eq:fact:branch-values-bound-1}
  \lout{s} = \lout{t} + 1
  \end{equation}
  and in particular the first inequality of the claim.

  Let us now prove that $\lout{s} \leq K$,
  where $K \eqdef 3 (\card{V} + 1) 4^{\card{V}+1}$.
  Suppose, towards a contradiction, that $\lout{s} > K$.
  Observe that $\lout{t} \geq K$ because of
  \cref{eq:fact:branch-values-bound-1}.
  Let us consider the subtrees on the right of the branch from $s$ to $t$.
  The main idea of the proof is to replace these subtrees by smaller ones
  using \cref{thm:derivewitness}.
  Formally,
  let $U = \{p1 \in T \mid s \prefix p \pprefix t \wedge p1 \not\prefix t\}$.
  The set $U$ collects the right-children of the main branch from $s$ to the parent of $t$,
\gsnote{Introduce earlier in this section the ``left'' and ``right'' notions that we use in informal explanations?}
  excluding those that are on the branch themselves.
  Let $u_1, \ldots, u_k$ denote the elements of $U$,
  in lexicographic order,
  and let $S_i = \lsym{u_i}$ for $i = 1, \ldots, k$.
  Note that $S_i \in V$ since $G$ is in Chomsky normal form by assumption.
  Observe also that $\lout{s} \leq \lout{t} + \displ{\#_1 \cdots \#_k}$
\gsnote{This uses displacements $\displ[^G]{w}$ which are not defined (only $\displ[^G]{}$).}
  due to the flow conditions.
  It follows from $\lout{s} = \lout{t} + 1$ that $\displ{\#_1 \cdots \#_k} > 0$.

  If the total size of the subtrees rooted in $u_1, \ldots, u_k$ is
  at most $K$,
  then the flow conditions entail that
  the input and output values of their nodes are all strictly positive,
  since $\lout{s} > K$ and $A = \{-1,0,1\}$ by assumption.
  The same holds for the output values of the nodes $p$ with
  $s \prefix p \prefix t$.
  So we may decrease all these values by one,
  retaining the flow conditions.
  Indeed,
  \cref{fact:dropped-values-1} guarantees that the first flow condition
  still holds for the parent of $s$.
  We obtain, in this way,
  a certificate of strictly smaller rank than $(\mathcal{T}, s, t)$.

  Otherwise,
  the total size of the subtrees rooted in $u_1, \ldots, u_k$ is
  at least $K + 1$.
  Observe that $k \leq \card{V} + 1$ by \cref{fact:branch-depth-bound}.
  According to \cref{thm:derivewitness},
  there exists $T_1, \ldots, T_k$,
  where each $T_i$ is a complete parse tree for $G[S_i]$ with
  yield $z_i$,
  such that $\card{T_1} + \cdots + \card{T_k} \leq 3 k 4^{\card{V}+1} \leq K$
  and $\sum z_1\ldots z_k>0$.
  Let us replace the subtrees rooted in $u_1, \ldots, u_k$
  by $T_1, \ldots, T_k$, respectively.
  Since $\lout{t} \geq K$,
  this induces a complete flow tree
  $\mathcal{T}' = (T', \lsymoperator['], \linoperator['], \loutoperator['])$,
  with $\lout[']{t} = \lout{t}$, and
  satisfying
  $\lout[']{s} = \lout[']{t} + \sum z_1\ldots z_k > \lout[']{t}$.
  The new output value of $s$ might be smaller,
  but \cref{fact:dropped-values-1} guarantees that the first flow condition
  still holds for the parent of $s$.
  The input values of $s$ and $t$ were not changed,
  so we have $\lin[']{s} = \lin{s} \leq \lin{t} = \lin[']{t}$.
  Therefore,
  $(\mathcal{T}', s, t)$ is a certificate
  of strictly smaller rank than $(\mathcal{T}, s, t)$.

  In both cases,
  we obtain a contradiction with the minimality of $(\mathcal{T}, s, t)$.
  The observation that $K \leq 6 \card{V} \cdot 4^{\card{V}+1}$
  concludes the proof of the fact.
  \qed
\end{proof}

\FactBranchValuesBoundThree*
\begin{proof}
  We first show that $\lin{s} \leq \lin{t} \leq \lin{s} + 1$.
  Recall that $\lin{s} \leq \lin{t}$ by definition of certificates.
  If we had $\lin{s} + 1 < \lin{t}$ then we could
  decrease the input value of $t$ by one,
  retaining the flow conditions by \cref{fact:dropped-values-2},
  and get a certificate
  of strictly smaller rank than $(\mathcal{T}, s, t)$,
  contradicting the minimality of $(\mathcal{T}, s, t)$.
  Therefore, $\lin{t} \leq \lin{s} + 1$.

  Observe that $t$ is an internal node since $\lsym{s} = \lsym{t}$
  cannot be in $A \cup \{\varepsilon\}$.
  Let us bound the input value of its first child.
  According to \cref{fact:traversal-inequation},
  it holds that $\lin{tj} \leq \lout{tj} + 2^{\card{V}}$
  for each child $tj$ of $t$.
  Let $k \in \{1, 2\}$ denote the number of children of $t$.
  The flow conditions together with
  the minimality of $(\mathcal{T}, s, t)$ guarantee that
  $\lin{t0} = \lin{t}$,
  $\lout{t} = \lout{t(k-1)}$, and
  $\lin{t(j+1)} = \lout{tj}$ for every $j = 0, \ldots, k-1$.
  We derive that $\lin{t0} \leq \lout{t} + 2 \cdot 2^{\card{V}}$.
  It follows from \cref{fact:branch-values-bound-1} that
  $\lin{t0} \leq K + 2^{\card{V} + 1}$,
  where $K \eqdef 6 \card{V} \cdot 4^{\card{V}+1}$.

  Let us now prove that $\lin{t} \leq H$,
  where $H \eqdef K + 2^{\card{V} + 1}$.
  The proof is similar to the proof of \cref{fact:branch-values-bound-1}.
  Suppose, towards a contradiction, that $\lin{t} > H$.
  Observe that $\lin{s} \geq H$.
  Let us consider the subtrees on the left of the branch from $s$ to $t$.
  The main idea of the proof is to replace these subtrees by smaller ones
  using \cref{thm:derivewitness}.
  Formally,
  let $U = \{p0 \in T \mid s \prefix p \pprefix t \wedge p0 \not\prefix t\}$.
  The set $U$ collects the left-children of the main branch from $s$ to the parent of $t$,
  excluding those that are on the branch themselves.
  Let $u_1, \ldots, u_k$ denote the elements of $U$,
  in lexicographic order,
  and let $S_i = \lsym{u_i}$ for $i = 1, \ldots, k$.
  Note that $S_i \in V$ since $G$ is in Chomsky normal form by assumption.
  Observe also that $\lin{t} \leq \lin{s} + \displ{\#_1 \cdots \#_k}$.
\gsnote{This uses displacements $\displ[^G]{w}$ which are not defined (only $\displ[^G]{}$).}
  It follows that
  $\displ{\#_1 \cdots \#_k} \geq 0$ and that
  $\displ{\#_1 \cdots \#_k} > 0$ if $\lin{s} < \lin{t}$.

  If the total size of the subtrees rooted in $u_1, \ldots, u_k$ is
  at most $K$,
  then the flow conditions entail that
  the input and input values of their nodes are all strictly positive,
  since $\lin{t} > H \geq K$ and $A = \{-1,0,1\}$ by assumption.
  The same holds for the input values of the nodes $p$ with
  $s \prefix p \prefix t$.
  So we may decrease all these values by one,
  retaining the flow conditions.
  Indeed,
  the first flow condition still holds for $t$
  since $\lin{t0} \leq H$.
  We obtain, in this way,
  a certificate of strictly smaller rank than $(\mathcal{T}, s, t)$.

  Otherwise,
  the total size of the subtrees rooted in $u_1, \ldots, u_k$ is
  at least $K + 1$.
  Observe that $k \leq \card{V} + 1$ by \cref{fact:branch-depth-bound}.
  According to \cref{thm:derivewitness},
  there exists $T_1, \ldots, T_k$,
  where each $T_i$ is a complete parse tree for $G[S_i]$ with
  yield $z_i$,
  such that $\card{T_1} + \cdots + \card{T_k} \leq 3 k 4^{\card{V}+1} \leq K$
  and $\sum z_1\ldots z_k \geq 0$.
  Moreover $\sum z_1\ldots z_k = 0$ only if $\lin{s} = \lin{t}$.
  Let us replace the subtrees rooted in $u_1, \ldots, u_k$
  by $T_1, \ldots, T_k$, respectively.
  Since $\lin{s} \geq H \geq K$,
  this induces a complete flow tree
  $\mathcal{T}' = (T', \lsymoperator['], \linoperator['], \loutoperator['])$,
  with $\lin[']{s} = \lin{s}$, and
  satisfying
  $\lin[']{t} = \lin[']{s} + \sum z_1\ldots z_k \geq \lin[']{s}$.
  The first flow condition still holds for $t$ since
  $\lin[']{t0} = \lin{t0} \leq H \leq \lin{s} = \lin[']{s} \leq \lin[']{t}$.
  The output values of $s$ and $t$ were not changed.
  It follows that $\lin[']{s} < \lin[']{t}$ or $\lout[']{t} < \lout[']{s}$.
  Indeed,
  if $\lin[']{s} = \lin[']{t}$ then
  $\sum z_1\ldots z_k = 0$, hence, $\lin{s} = \lin{t}$,
  which entails that $\lout[']{t} = \lout{t} < \lout{s} = \lout[']{s}$.
  Therefore,
  $(\mathcal{T}', s, t)$ is a certificate
  of strictly smaller rank than $(\mathcal{T}, s, t)$.

  In both cases,
  we obtain a contradiction with the minimality of $(\mathcal{T}, s, t)$.
  The observation that $H \leq 7 \card{V} \cdot 4^{\card{V}+1}$
  concludes the proof of the fact.
  \qed
\end{proof}

\FactBranchValuesBoundTwo*
\begin{proof}
  Let us write $t = s j_1 \cdots j_k$ where each $j_i \in \{0, 1\}$,
  and let $p_i = s j_1 \cdots j_i$ for $i = 0, \ldots, k$.
  We first show that for every $0 < i \leq k$,
  \begin{equation}
  \label{eq:branch-values-bound-2}
  \lout{p_i} \leq \lout{p_{i-1}} + 2^{\card{V}}.
  \end{equation}
  Indeed,
  the flow conditions together with
  the minimality of $(\mathcal{T}, s, t)$ guarantee,
  for every ancestor $s \prefix p \pprefix t$,
  that $\lout{p} = \lout{p1}$ and
  that $\lout{p0} = \lin{p1}$ if $p0 \prefix t$.
  In the latter case,
  $\lin{p1} \leq \lout{p1} + 2^{\card{V}}$
  by \cref{fact:traversal-inequation},
  hence,
  $\lout{p0} \leq \lout{p} + 2^{\card{V}}$.
  We have thus shown that
  $\lout{pj} \leq \lout{p} + 2^{\card{V}}$
  for every ancestor $s \prefix p \pprefix t$
  and every $j \in \{0, 1\}$
  such that $pj \prefix t$.

  We derive from \cref{eq:branch-values-bound-2} that
  $\lout{p_i} \leq \lout{s} + i 2^{\card{V}}$
  for all $0 < i < k$.
  Recall that $k \leq \card{V} + 1$ by \cref{fact:branch-depth-bound}.
  It follows from \cref{fact:branch-values-bound-1}
  that for every node $p$ such that $s \prefix p \prefix t$,
  we have
  \begin{equation*}
  \lout{p} \leq \lout{s} + \card{V}\cdot2^{\card{V}}
  \leq 6 \card{V} \cdot 4^{\card{V}+1} + \card{V}\cdot 2^{\card{V}}
  \leq 4^{2(\card{V}+1)}
  \end{equation*}
  According to \cref{fact:dropped-values-1},
  it holds that $\lout{p} = -\infty$ for every proper ancestor $p \pprefix s$,
  which concludes the proof of the fact.
  \qed
\end{proof}

\FactBranchValuesBoundFour*
\begin{proof}
  Let us write $t = j_1 \cdots j_k$ where each $j_i \in \{0, 1\}$,
  and let $p_i = j_1 \cdots j_i$ for $i = 1, \ldots, k$.
  We claim that
  $\lin{p_{i-1}} \leq \lin{p_i} + 2^{\card{V}}$
  for every $1 < i \leq k$.
  Indeed,
  the flow conditions together with
  the minimality of $(\mathcal{T}, s, t)$ guarantee,
  for every ancestor $\varepsilon \pprefix p \pprefix t$,
  that $\lin{p} = \lin{p0}$ and
  that $\lout{p0} = \lin{p1}$ if $p1 \prefix t$.
  In the latter case,
  $\lin{p0} \leq \lout{p0} + 2^{\card{V}}$
  by \cref{fact:traversal-inequation},
  hence,
  $\lin{p} \leq \lin{p1} + 2^{\card{V}}$.
  We have thus shown that
  $\lin{p} \leq \lin{pj} + 2^{\card{V}}$
  for every ancestor $\varepsilon \pprefix p \pprefix t$
  and every $j \in \{0, 1\}$
  such that $pj \prefix t$.
  This concludes the proof of the claim.

  Observe that $s = j_1 \cdots j_h$ where $h = \len{s}$.
  We derive from the claim that, firstly,
  $\lin{p_i} \leq \lin{s} + (h-i) 2^{\card{V}}$
  for all $0 < i < h$,
  and secondly,
  $\lin{p_i} \leq \lin{t} + (k-i) 2^{\card{V}}$
  for all $h < i < k$.
  Recall that $h \leq \card{V}$ and $(k-h) \leq \card{V} + 1$ by \cref{fact:branch-depth-bound}.
  It follows from \cref{fact:branch-values-bound-3}
  that,
  for every ancestor $\varepsilon \pprefix p \prefix t$,
  we have
  \begin{equation*}
  \lin{p}
  \leq \max\{\lin{s}, \lin{t}\} + \card{V} 2^{\card{V}}
  \leq 7 \card{V} \cdot 4^{\card{V}+1} + \card{V} 2^{\card{V}}
  \leq 4^{2(\card{V}+1)}
  \end{equation*}
  which concludes the proof of the fact.
  \qed
\end{proof}

\end{document}